\documentclass{article}
\usepackage[colorlinks=true, allcolors=blue]{hyperref}
\usepackage[utf8]{inputenc}
\usepackage{graphicx}
\usepackage{enumitem}
\usepackage{xcolor}
\usepackage{fullpage}
\usepackage{authblk}
\usepackage{caption}

\usepackage{amsmath, amsfonts, amsthm, amssymb, dsfont}
\usepackage{thmtools, thm-restate}
\usepackage{braket}

\declaretheorem{theorem}

\declaretheorem{corollary}

\theoremstyle{definition}
\newtheorem{definition}{Definition}
\newtheorem*{definition*}{Definition}

\title{Limitations of measure-first protocols in quantum machine learning}

\author[1]{Casper Gyurik 
\thanks{\href{mailto:c.f.s.gyurik@liacs.leidenuniv.nl}{c.f.s.gyurik@liacs.leidenuniv.nl}}}
\author[1]{Riccardo Molteni
\thanks{\href{mailto:r.molteni@liacs.leidenuniv.nl}{r.molteni@liacs.leidenuniv.nl}}}
\author[1]{Vedran Dunjko
\thanks{\href{mailto:v.dunjko@liacs.leidenuniv.nl}{v.dunjko@liacs.leidenuniv.nl}}}
\affil[1]{{\small applied Quantum algorithms (aQa), Leiden University, The Netherlands}} 

\date{}

\begin{document}

\maketitle

\begin{abstract}

In recent works, much progress has been made with regards to so-called randomized measurement strategies, which include the famous methods of classical shadows and shadow tomography. 
In such strategies, unknown quantum states are first measured (or “learned”), to obtain classical data that can be used to later infer (or “predict”) some desired properties of the quantum states. 
Even if the used measurement procedure is fixed, surprisingly, estimations of an exponential number of vastly different quantities can be obtained from a polynomial amount of measurement data.
This raises the question of just how powerful ``measure-first'' strategies are, and in particular, if all quantum machine learning problems can be solved with a measure-first, analyze-later scheme. 
This paper explores the potential and limitations of these measure-first protocols in learning from quantum data. 
We study a natural supervised learning setting where quantum states constitute data points, and the labels stem from an unknown measurement. 
We examine two types of machine learning protocols: ``measure-first'' protocols, where all the quantum data is first measured using a fixed measurement strategy, and
``fully-quantum'' protocols where the measurements are adapted during the training process. 
Our main result is a proof of separation. 
We prove that there exist learning problems that can be efficiently learned by fully-quantum protocols but which require exponential resources for measure-first protocols. 
Moreover, we show that this separation persists even for quantum data that can be prepared by a polynomial-time quantum process, such as a polynomially-sized quantum circuit.
Our proofs combine methods from one-way communication complexity and pseudorandom quantum states. 
Our result underscores the role of quantum data processing in machine learning and highlights scenarios where quantum advantages appear.

\end{abstract}

\section{Introduction}
\label{sec:intro}

A central question in quantum machine learning revolves around understanding the various types of advantages one can achieve by exploiting quantum effects. 
Some of the most interesting scenarios arise when the dataset itself comprises quantum states, which can then be processed fully coherently, or through elaborate measurement strategies. 
In this context, exponential advantages have been identified when coherent measurements of multiple copies of a given quantum state are allowed~\cite{chen:memory, huang:science, huang:nature}. 
In a parallel related line, there have been significant breakthroughs in extracting useful classical information from quantum states using the versatile toolkit of randomized measurements~\cite{elben:review}. 
This toolkit includes the groundbreaking concept of classical shadows~\cite{huang:classicalshadows, huang:science}, which can extract an efficient classical description of quantum states that allows one to compute various physical properties. 
These two distinct research lines raises questions about the different capabilities of quantum machine learning protocols that employ coherent manipulation of quantum states and adaptive measurements compared to machine learning protocols that use a universal measurement strategy to extract valuable classical data. 
\textit{In particular, these results raise the natural question of whether it is possible that a ``measure-first'' protocol can be universally used as a substitute for any ``fully-quantum'' protocol in quantum machine learning tasks.}
We formally define what we mean by a ``measure-first'' or a ``fully-quantum'' protocol in Section~\ref{subsec:problem_models} (see Definition~\ref{def:fully-quantum} and Definition~\ref{def:measure-first}), and we provide an overview in Figure~\ref{fig:intro_overview}. 
Intuitively, in a measure-first protocol, one first measures the quantum state independently from the specific task you want to use the extracted information for later on  (i.e., ``measure first, ask later'').
In contrast, a fully quantum protocol performs full quantum processing of the input states, allowing the measurements to adapt to the particular instance of the learning problem.

In this paper, we shed light on the limitations of measure-first protocols involving universal measurement strategies by raising the question whether there exists learning problems for which no measure-first protocol is as powerful as a fully quantum one.
We remark that outside the domain of machine learning, affirmative responses to this question are already known, such as in (distributed) sampling tasks~\cite{montanaro:fourier} or in the context of relational problems~\cite{aaronson:qubit_coin}.
In this work instead, we study a natural quantum version  of the standard supervised learning setting~\cite{aimeur:quantum_data}, where the input consists of multiple copies of a quantum state and labels correspond to outcomes of some unknown measurement of the quantum states. 
In other words, each instance of the learning problem corresponds to a distinct, undisclosed measurement. 
The objective is to learn how to reproduce this measurement from the data in such a way that when presented with a new quantum state, the learning protocol can generate the correct outcome with the right probabilities.
In this setting, we prove that there exists a learning problem for which there is an exponential separation in the training data required by the two protocols to successfully complete the task.

\begin{figure*}[h!]
\centering
\includegraphics[width=0.85\linewidth]{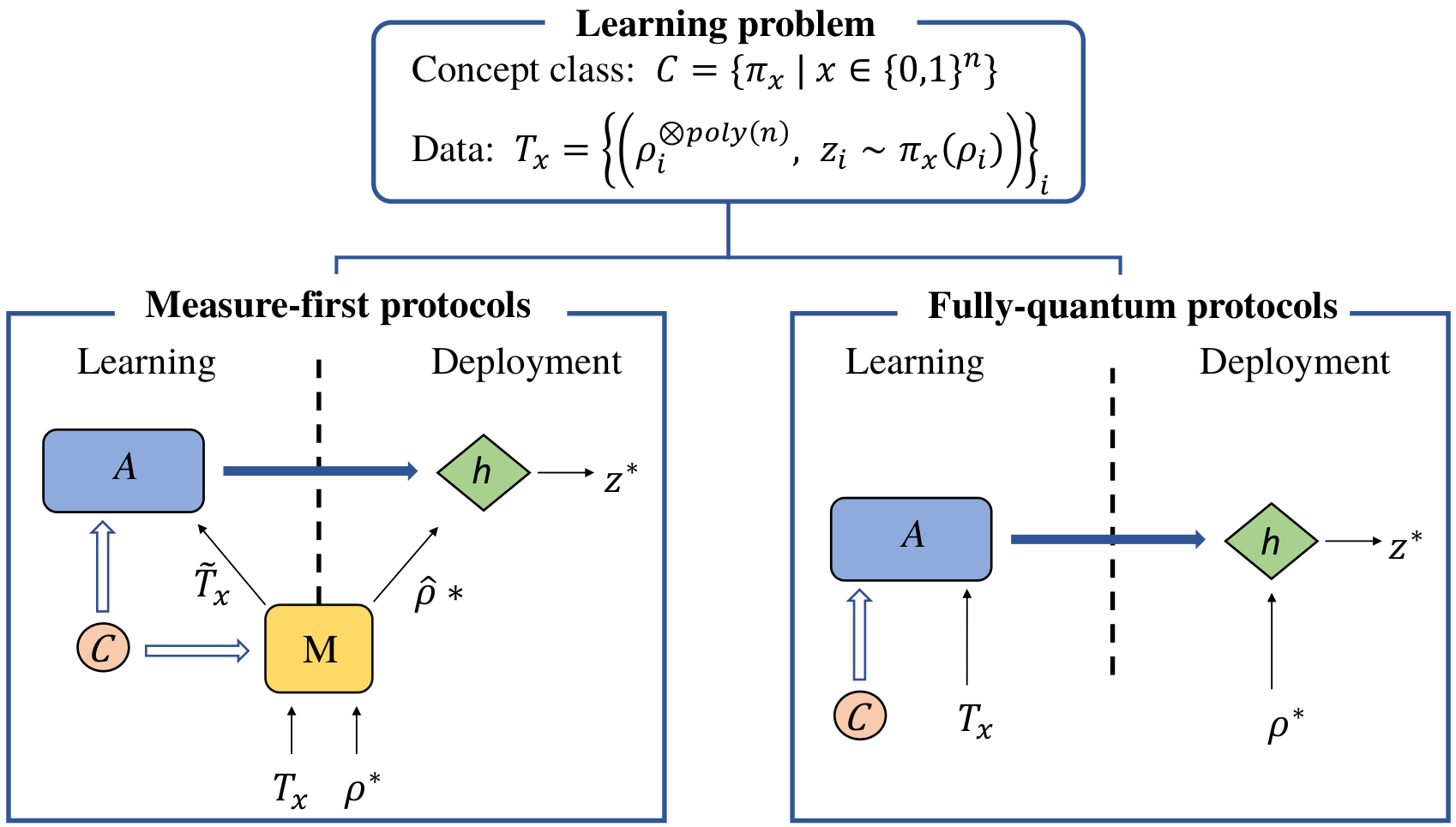}
\caption{An overview of the quantum machine learning protocols explored in our paper. 
In both protocols, the quantum algorithm $A$ is tasked with processing training data $T_x$ to output a classical description of a function $h$ that correctly produces samples $z^*$ for new input states $\rho^*$ (i.e., $h$ should implement the quantum measurement). 
The distinction between the two protocols lies in fact that in the \textit{measure-first} protocol the quantum input states undergo a randomized measurement strategy $M$, which converts a quantum state $\rho^*$ into a classical representation $\hat{\rho}^*$.
Importantly, this strategy is allowed to depend only on the concept class $\mathcal{C}$ (which captures the general learning problem), not on the specific target concept $\pi_x \in \mathcal{C}$. 
In other words, the same randomized measurement strategy $M$ is used to preprocess the quantum data for all concepts $\pi_x\in \mathcal{C}$.} 

\label{fig:intro_overview}
\end{figure*}

\subsection{High-level overview of learning setting and main result}

The machine learning problem concerns learning an unknown measurement acting on a set of input quantum states. 
Specifically, the data that the learning protocol gets consists of pairs of copies of $n$-qubit quantum states $\rho$ drawn from some distribution $\mathcal{D}$ together with a corresponding label $z\in\{0,1\}^{2n+1}$. 
The first $n$ bits of $z$ encode a $2^n$-outcome POVM measurement $\Lambda_x = \{E^x_z \mid z \in \{0,1\}^n\}$ from some set $\Lambda = \{\Lambda_x \mid x \in \{0,1\}^n\}$. 
The remaining $n+1$ bits are determined by the outcome of $\Lambda_x$ on the quantum state $\rho$. 
The goal of the machine learning protocol is to ``learn'' how to reproduce the measurement $\Lambda_x$. 
More precisely, the trained learning protocol has to receive as input an unseen quantum state $\rho'$ and output a sample $z'$ in agreement with the probability distribution $\pi_x(\rho)$, where $\pi_x(\rho)$ is the distribution of the measurement outcomes of $\Lambda_x$.

This paper explores whether solving learning problems such as the above requires a quantum computer capable of adaptive measurements on the training data, or if a fixed measurement strategy that produces classical representations of the quantum states is sufficient.
Specifically, we consider so-called ``\textit{measure-first protocols}'' that are forced to measure the input states and use the obtained classical description to train the machine learning model to produce samples from the target distribution. 
Importantly, in a measure-first protocol the measurements are not allowed to depend on the training data, but are unconstrained otherwise.
On the other hand, we consider so-called ``\textit{fully-quantum protocols}'' that can coherently process the quantum states and adjust the measurements based on the training data. 
Our main result is the existence of a learning problem where the quantum data is efficiently generatable on a quantum computer for which a ``measure-first protocol'' requires an exponential amount of data to be able to produce measurement outcomes on new input states, whereas a fully-quantum protocol only requires a polynomial amount of data. 
Additionally, we require that the protocols must be efficient in the sense that they run in time polynomial in $n$.
Using the notion of learning we presented above, we now give an informal description of the main result of this paper. 

\begin{theorem}
\textnormal{(Informal)} There exist a concept class $\mathcal{C}$, defined by a set of measurements and a distribution over quantum states $\mathcal{D}$ such that a ``measure-first'' protocol cannot learn $\mathcal{C}$ with a polynomial amount of training data. 
On the other hand, there exists a ``fully-quantum'' protocol which can learn $\mathcal{C}$ efficienctly with respect to both sample and time complexity .  
\end{theorem}

This result suggests that the number of bits required to store $n$ qubits, such that it can be used by a learner to successfully solve the learning problem, is exponential in $n$. 
Such a conclusion is analogous to what Montanaro refers as ``anti-Holevo'' theorems~\cite{montanaro:fourier}.

\subsection{Related work}
\label{subsec:related_work}

In this section we discuss related works and highlight their relationships to our learning setting.
Firstly, in~\cite{huang:classicalshadows}, the authors introduced a randomized measurement technique tailored to extract a classical description of a quantum state $\rho$. 
This description enables the computation of the expectation values of any set of observables $\{O_k\}_{k=1}^M$ -- provided they have a low ``shadow norm'', such as when the observables are local -- up to a precision of $\epsilon$.
Notably, they showed that a number of copies of $\rho$, scaling logarithmically with the number of observables $M$ and inverse-polynomially in the precision $\epsilon$, suffices for this task. 
Directly applying these techniques to our learning setting that is focused on learning a $2^n$-outcome POVM by estimating the probability of each possible measurement outcome $z \in \{1,2,...,2^n\}$ requires exponential precision $\epsilon$.

The concept of shadow tomography, introduced in~\cite{aaronson:shadowtomography}, revolves around the problem of computing the expectation values of any set of $M$ two-outcome measurements on an $n$-qubit state $\rho$ up to precision $\epsilon$. 
It has been shown that this can be done using a number of copies of $\rho$ that scale polylogarithmically in $M$, linearly in $n$, and inverse-polynomially in $\epsilon$~\cite{aaronson:shadowtomography}. 
In contrast to the methods in~\cite{huang:classicalshadows}, the approach in~\cite{aaronson:shadowtomography} requires coherent measurements on multiple copies of $\rho$. 
Additionally, as demonstrated in~\cite{huang:science, chen:exponential} for specific tasks, the capacity to coherently measure multiple copies of quantum states provides an exponential advantage in sample complexity over sequential measurements. 
Considering our framework, where measurements can act coherently on multiple copies of each input state, one might question whether such strategies enable a measure-first protocol to solve our learning task. 
However, it is important to note that coherent measurements do not improve the scaling with respect to the precision $\epsilon$, leaving room for the possibility of a separation.

In~\cite{aaronson:distributions}, the authors extended the concept of shadow tomography to the scenario of learning a $K$-outcome POVM (for $K\geq2$) selected from a set of $M$ unknown quantum measurements. 
In contrast to the binary outcome case, the goal now is to approximate an unknown distribution up to precision $\epsilon$ in total variation distance, rather than focusing on expectation values.
Their procedure requires a number of copies of the quantum state scaling linearly with $K$ and $n$, polylogarithmically with $M$, and inverse-polynomially with the precision $\epsilon$. 
Moreover, they establish the optimality of this scaling with respect to the dependence on the number of outcomes $K$.
This result prompts the question of whether our separation between measure-first and fully-quantum protocols can be directly inferred from it. 
However, in establishing their lower bound, it is important to note that no assumptions were made regarding the complexity of the unknown quantum state. 
The crux of our study lies in demonstrating that the measure-first protocol falls short of replicating the unknown measurement, even on quantum states that are efficiently preparable. 
It is moreover important to highlight that while their shadow tomography procedures can be employed to construct a measure-first protocol by ``shadowfying'' input states to approximate the expected values of each measurement outcome, this is not strictly necessary for our task. 
Specifically, understanding the probability of each outcome allows for the creation of an ``evaluator'' that can compute the correct probability for every outcome. 
However, to resolve our learning problem, a ``generator'' (i.e., an algorithm generating samples with the correct probabilities) already suffices, and it does not necessarily require computing the output probabilities~\cite{sweke:learnability}.

In~\cite{cheng:learnability}, the authors provide upper bounds for the dual problem of shadow tomography, or more specifically the problem of learning a measurement. 
In particular, they studied the task of learning an unknown two-outcome POVM denoted $\{E, I - E\}$, from data of the form $\{(\rho_i, \text{Tr}(E \rho_i))\}_{i=1}^{m}$. 
They showed that to approximate the unknown $E$ up to a precision of $\epsilon$ on new $n$-qubit input states, $\mathcal{O}(2^n/\epsilon^2)$ training samples are sufficient. We remark, however, that the number of required samples scales exponentially with the number of qubits.

Finally, in~\cite{jerbi:power} the authors study quantum process learning, where the task is to learn an unknown unitary $U$, from data of the form $\{\ket{\psi_i},U\ket{\psi_i}\}$. 
In particular, they study the limitations of what they call incoherent learning, where the learner is constraint to first measure multiple copies of the data $U\ket{\psi_i}$. 
While they therefore also study the problem of extracting classical information from quantum data and utilizing it in the learning process, the setting in their work differs from ours. Namely, the quantum states in our scenario are labeled by a sample obtained from the unknown measurement process, whereas in~\cite{jerbi:power} the labels are the input quantum states when evolved under some unknown target unitary.


\section{Main result}
\label{sec:main_result}

In this section, we present the key findings of our paper. 
We begin in Section~\ref{subsec:problem_models} by defining the learning problem we study and we introduce the two types of learning models we analyze: one involving a universal randomized measurement strategy (i.e., ``\textit{measure-first quantum machine learning}''), and the other using adaptive measurements that are trained separately on each problem instance (i.e., ``\textit{fully-quantum machine learning}''). 
Afterwards, in Section~\ref{subsubsec:learnability}, we show how the fully-quantum machine learning model can efficiently solve our learning problem. 
In Section~\ref{subsec:unlearnability}, we present our first main result showing that no measure-first quantum machine learning model can solve our learning problem efficiently.
Finally, in Section~\ref{subsec:unlearnability}, we show that this separation between the models still holds if the quantum states in the data are efficiently preparable.

\subsection{The learning problem and learning models}
\label{subsec:problem_models}

The learning problem we study is the learning of a measurement. 
In particular, it involves generating samples from a distribution induced by measuring an (unknown) POVM measurement $\Lambda_x$ on $n$-qubit quantum states. 
The (unknown) target measurement $\Lambda_x$ is drawn from a set of POVMs $\{\Lambda_x \mid x \in \{0,1\}^n\}$, and each measurement $\Lambda_x$ is a computational basis measurement preceded by an $n$-qubit unitary $U_x$, i.e., 
\[
\Lambda_x=\left\{E^x_z \big|\text{  }z = 1, \dots, 2^n\right\}, \text{ where } E^x_z= U_x\ket{z}\bra{z}U^\dagger_x.
\]
During training the learner is given a set of examples $T_x$, where each example consists of a polynomial number of copies of a phase state $\ket{\psi_{f'}}$ together with a sample from the associated POVM-induced distribution. 
This is a special case of labeled quantum data, which was introduced in~\cite{aimeur:quantum_data}, where we are additionally allowed to have access to polynomially many copies of the input quantum state.
We formalize our learning problem by generalizing the standard PAC learning framework~\cite{valiant:pac}.
In our generalization, a concept corresponds to a quantum randomized function, i.e., a function that on each quantum input state outputs a sample from a random variable (which in our case corresponds to the outcomes of a POVM on the input quantum state). 
Before we define the concept class studied throughout this paper, we first setup some auxilliary definitions.
\begin{definition}[Auxiliary definitions/notation]\label{def:aux}
\hspace{0pt}
\begin{itemize}
    \item Let $N = 2^n$, or equivalently $n = \log_2 N$. 
    \item We identify a function $f:\{0,1\}^n \rightarrow \{0,1\}$ with its truth table $f\in\{0,1\}^N$, and we denote its corresponding \textit{phase state} with 
    \begin{align}
        \label{eq:phase_state}
        \ket{\psi_f} = \frac{1}{\sqrt{N}}\sum_{i=1}^{N}(-1)^{f_i}\ket{i}.
    \end{align}
    \item Let $\mathcal{S}_{\mathrm{phase}} = \{\ket{\psi_f}^{\otimes \ell} \mid f \in \{0,1\}^N\}$ denote the set of $\ell = \mathrm{poly}(n)$ copies of $n$-qubit phase states.
    \item We write $x \sim \pi$ to denote that $x$ was drawn according to a distribution $\pi$.
    \item We write $\mathcal{U}(\mathcal{X})$ for the uniform distribution over a set $\mathcal{X}$.
    \item We write $\Delta(\mathcal{X})$ for the set of all distributions over a set $\mathcal{X}$.
\end{itemize}
\end{definition}

\begin{definition}[Concept class]\label{def:concepts}
We define our concept class as $\mathcal{C} = \{ \pi_x \mid x \in \{0, 1\}^n\}$, such that
\begin{equation}
\begin{split}
 \label{eq:concepts}
    \pi_x: \mathcal{S}_{\mathrm{phase}} &\rightarrow \Delta(\{0,1\}^{2n+1})\\
    \ket{\psi_f}^{\otimes \ell}&\mapsto \pi_x(f)
    \end{split}
    \end{equation}
    where $\pi_x(f)$ is a distribution over samples $(x, y, b)$, where $(y,b)\sim  \mathcal{U}(R_{f}(x))$ and 
    \begin{equation}
    \label{eq:relation}
    R_f(x) = \{(x, y, b) \mid x,y \in \{0,1\}^n,\text{ }b \in \{0,1\},\text{ }f(y) \oplus f(y \oplus x) = b\}.
    \end{equation}
\end{definition}

In particular, $\pi_x$ is a randomized function which takes as input a polynomial number of copies of a phase state and outputs a sample $z$ consisting of $x \in \{0,1\}^n$ together with some $(y,b) \in \{0,1\}^{n+1}$ drawn from the uniform distribution over $R_f(x)$. 
Importantly, in~\cite{aaronson:qubit_coin} the authors showed that for each $x \in \{0,1\}^n$ there exist a POVM measurement $\Lambda_x$ which when applied to a phase state $\ket{\psi_f}$ outputs a pair exactly satisfying the relation $R_f(x)$. 
With regards to our learning problem, the task of the learning protocols is to learn this measurement.
In short, a learner is given several evaluations of the randomized function $\pi_x$ in the form of training data and its objective is to implement a randomized function $\widetilde{\pi}_x$ that closely approximates $\pi_x$ on most input states. 
In this paper, we compare two categories of machine learning systems that can tackle problems of this type. 
First, we introduce what we call a ``\textit{fully-quantum protocol}''.

\begin{definition}[Fully-quantum protocol]
\label{def:fully-quantum} 
A \textit{fully-quantum protocol} for the concept class $\mathcal{C}$ in Definition~\ref{def:concepts} is a polynomial-time quantum algorithm~$A$ that takes as input training data of the form 
\begin{align}
\label{eq:qdata}
    T_x = \left\{\left(\ket{\psi_{f^{(i)}}}^{\otimes \ell}, (x, y, b)\right) \; \big| \;(x, y, b)\sim\pi_x(f^{(i)})\text{ and } f^{(i)} \sim \mathcal{U}\left(\{0,1\}^N\right)  \right\}_{i=1}^{\text{poly}(n)},
\end{align}
and outputs a \textit{classical description of a polynomial-time quantum algorithm} that on input $\ket{\psi_f}^{\otimes \ell} \in \mathcal{S}_{\mathrm{phase}}$ generates a sample from a distribution $\widetilde{\pi}_x(f) \in \Delta(\{0,1\}^{2n+1})$. 
\end{definition}

We emphasize that for a ``fully quantum'' protocol, the learning algorithm must produce a \textit{classical description} of the quantum algorithm generating samples from $\widetilde{\pi}_x$. 
Consequently, we do not store any quantum states from the training data in quantum memory, which would be more general but not studied in this paper.
Ultimately, the goal of the protocol is to implement a randomized function $\widetilde{\pi}_x$ that closely approximates the actual data-generating randomized function $\pi_x$ for most of the input quantum states.

\begin{definition}[$(\epsilon, \delta, p_{\mathrm{succ}})$-fully-quantum learnable]
\label{def:fully-quantum_eff}
We say that $\mathcal{C}$ is \textit{$(\epsilon, \delta, p_{\mathrm{succ}})$-fully-quantum learnable} if there exists a fully-quantum protocol $A$  such that for every $\pi_x \in \mathcal{C}$, with probability at least $p_{\mathrm{succ}}$ we have
\begin{align}
\label{eq:success_criteria_q}
\mathsf{Pr}_{f \sim_U \{0,1\}^N}\left(|| \widetilde{\pi}_x\left(f\right) - \pi_x(f) ||_{TV} \leq 1-\epsilon \right) \geq \delta,
\end{align}
where $\widetilde{\pi}_x(f)\in\Delta(\left\{0,1\right\}^{2n+1})$ denotes the distribution that the polynomial-time quantum algorithm obtained from the learning algorithm $A$ generates samples from on input $\ket{\psi_f}^{\otimes \ell} \in \mathcal{S}_{\mathrm{phase}}$.
\end{definition}

Next, we introduce a ``\textit{measure-first protocol}'' which consists of two components: (i) a randomized measurement strategy $M$, and (ii) a learning algorithm $A$.
The main difference between a measure-first protocol and a fully-quantum protocol is that the former involves a randomized measurement procedure that first measures the quantum states before putting it into a learning algorithm. 
Importantly, the measure-first protocol is allowed to perform arbitrary coherent measurements on all input quantum states (i.e., the polynomially-many copies of the phase states). 
The only constraint is that the \textit{measurement strategy cannot depend on the specific target concept} of the learning problem.
In short, a randomized measurement strategy is a polynomial-time algorithms that maps a polynomial number of copies of a phase state $\ket{\psi_f}$ to some classical description $\widehat{\psi}_f \in \{0,1\}^{m}$ for some $m = \mathrm{poly}(n)$. 
These classical descriptions $\widehat{\psi}_f$ are then used as input for the learning algorithm, that is tasked with implementing a randomized function close to $\pi_x$. 

\begin{definition}[Measure-first protocol]
\label{def:measure-first}
A \textit{measure-first protocol} is a tuple $(M, A)$ where
    \begin{itemize}
        \item $M$ is a measurement strategy that in time $\mathcal{O}(\mathrm{poly}(n))$ maps $\ket{\psi_f}^{\otimes \ell} \in \mathcal{S}_{\mathrm{phase}}$ to some $\widehat{\psi}_{f}\in \{0,1\}^{m}$, where $m = \mathrm{poly}(n)$.
        \item $A$ is a polynomial-time quantum algorithm that takes input of the form
        \begin{align}
            \label{eq:cdata}
                T_x^M = \left\{\left(\widehat{\psi}_{f^{(i)}}, (x, y, b)\right) \; \big| \;(x, y,b)\sim\pi_x(f^{(i)})\text{ and } f^{(i)} \sim \mathcal{U}\left(\{0,1\}^N\right)  \right\}_{i=1}^{\text{poly}(n)},
        \end{align}
        and outputs a \textit{description of a polynomial-time quantum algorithm} that on input $\widehat{\psi}_f = M(\ket{\psi_f}^{\otimes \ell})$ generates a sample from a distribution $\widetilde{\pi}_x(f) \in \Delta(\{0,1\}^{2n+1})$. 
    \end{itemize}
\end{definition}

Note that the distinction between measure-first and fully-quantum protocols lies in Eq.~\eqref{eq:cdata}, where the data is measured instead of remaining quantum states. 
Nonetheless, the measurement strategy $M$ is entirely arbitrary and fully unrestricted.
Recall that the objective of the protocol is to implement a randomized function $\widetilde{\pi}_x$ that closely approximates the actual data-generating randomized function $\pi_x$ on most inputs.

\begin{definition}[$(\epsilon, \delta, p_{\mathrm{succ}})$-measure-first learnable]
\label{def:measure-first_learnable}
We say that $\mathcal{C}$ is \textit{$(\epsilon, \delta, p_{\mathrm{succ}})$-measure-first learnable} if there exists a measure-first protocol $(M, A)$  such that for every $\pi_x \in \mathcal{C}$, with probability at least $p_{\mathrm{succ}}$ we have
\begin{align}
\label{eq:success_criteria_c}
\mathsf{Pr}_{f \sim_U \{0,1\}^N}\left(|| \widetilde{\pi}_x(f) - \pi_x(f) ||_{TV} \leq 1-\epsilon \right) \geq \delta,
\end{align}
where $\widetilde{\pi}_x(f)\in\Delta(\left\{0,1\right\}^{2n+1})$ denotes the distribution that the polynomial-time quantum algorithm obtained from the learning algorithm $A$ generates samples from on input $\widehat{\psi}_f = M(\ket{\psi_f}^{\otimes \ell})$.
\end{definition}

\subsection{Fully-quantum learnability}
\label{subsubsec:learnability}

In this section we describe how the concept class in Definition~\ref{def:concepts} is fully-quantum learnable.

\begin{restatable}{proposition}{fullyquantum}
\label{prop:fully-quantum_learnable}
    The concept class in Definition~\ref{def:concepts} is $(0, 0, 1)$-fully-quantum learnable.
\end{restatable}

The proof of Proposition~\ref{prop:fully-quantum_learnable} can be found in Appendix~\ref{appendix:proof_fully-quantum_learnable}, and we provide a high-level overview of the fully-quantum protocol here.
Firstly, the fully-quantum protocol reads out $x$ from one of the samples generated by $\pi_x$ in the training data. 
Next, a quantum circuit denoted as $U_x$ is constructed as outlined in~\cite{aaronson:qubit_coin}, which when measuring $U_x\ket{\psi_f}$ in the computational basis generates a sample from $\mathcal{U}(R_f(x))$.
Crucially, it is worth noting that these quantum circuits $U_x$ are of size $\mathcal{O}\left(\mathrm{poly}(n)\right)$ and can be constructed in time $\mathcal{O}\left(\mathrm{poly}(n)\right)$. 

It might seem that little genuine learning occurs when $x$ can be readily read out from a single example in $T_x$. 
However, we can introduce various levels of learning by providing only partial information about $x$ within the examples. 
This partial information should allow the recovery of $x$ from a polynomial number of examples. 
Several examples illustrating this are discussed in more detail in Appendix~\ref{appendix:proof_fully-quantum_learnable}.

\subsection{Limitations of measure-first protocols on general quantum states}
\label{subsec:unlearnability}

In the last section, we discussed how the concept class in Definition~\ref{def:concepts} is fully-quantum learnable.
Conversely, in this section we discuss our main result which states that this concept class is \textit{not} measure-first learnable.

\begin{restatable}{theorem}{measurefirst}
\label{thm:measure-first_nonlearnable}
The concept class in Definition~\ref{def:concepts} is not $(\epsilon, \delta, p_{\mathrm{succ}})$-measure-first learnable for $\epsilon\cdot\delta> 7/8$ and any $p_{\mathrm{succ}} > 0$.
\end{restatable}

The proof of Theorem~\ref{thm:measure-first_nonlearnable} is provided below, and we first present a concise overview of the proof here.
At its core, the proof hinges on the notion that the existence of a measure-first protocol for the concept class described in Definition~\ref{def:concepts} implies the existence of an efficient classical one-way communication protocol for the Hidden Matching (HM) problem~\cite{bar:hm}.
Notably, in~\cite{bar:hm}, it has been shown that the HM problem cannot be solved with a communication cost of $\mathcal{O}\left(\mathrm{poly}(n)\right)$ bits, even on a $7/8$ fraction of possible inputs.
In essence, one of the two parties can employ the measurement strategy to encode their input for the HM problem, transmit it to the other party, who can then utilize the learning algorithm $A$ to successfully solve the HM problem.
Intuitively, the reason behind why measure-first learning fails in that due to~\cite{bar:hm} it is not possible to compress a phase state $\ket{\psi_f}$ into a polynomially-sized classical representation $\widehat{\psi}_f$ that contains enough information to allow one to generate samples from the distributions $\pi_x(f)$ for all possible $x$.



\begin{proof}[Proof of Theorem~\ref{thm:measure-first_nonlearnable}]

The main building block of our proof of Theorem~\ref{thm:measure-first_nonlearnable} is a result in \emph{one-way communication complexity} by Bar-Yossef, Jayram and Kerenidis~\cite{bar:hm}.
They define a problem called \emph{Hidden Matching} (HM). 
Here Alice is given a string $f \in \{0,1\}^N$, while Bob is given a perfect matching $M$ on the set $[N]$, consisting of $N/2$ edges.
Bob's goal is to output some $(i, j, f_i \oplus f_j)$ for some edge $(i, j) \in M$. 
Their main result is:
\begin{theorem}[Classical hardness of HM~\cite{bar:hm}]
    \label{thm:hm}
    Let $\mathcal{M}$ be any set of perfect matchings on $[N]$ that is pairwise edge-disjoint and satisfies $|\mathcal{M}| = \Omega(N)$. 
    Let $\mu$ be the distribution over inputs to HM in which Alice's input is uniform in $\{0,1\}^N$ and Bob's input is uniform in $\mathcal{M}$.
    Then, any deterministic one-way protocol for HM that errs with probability at most $1/8$ with respect to $\mu$ requires $\Omega(\sqrt{N})$ bits of communication.
\end{theorem}

Suppose the concept class in Definition~\ref{def:concepts} is  $(\epsilon, \delta, p_{\mathrm{succ}})$-measure-first learnable using a measure-first protocol given by $(M, A)$ with $\epsilon\cdot\delta > 7/8$ and $p_{\mathrm{succ}} > 0$.
Throughout the proof, we will show that the existence of such a measure-first learning protocol contradicts the classical hardness of HM outlined in Theorem~\ref{thm:hm}.
To do so, consider the HM problem with $\mathcal{M} = \{M_x \mid x \in \{0,1\}^n\}$, where
\begin{align}
\label{eq:matchings}
    M_x = \{(y, y\oplus x) \mid y \in \{0,1\}^n \}.
\end{align}
and note that $|\mathcal{M}| = N$. 
To solve this instance of the HM problem Bob first generates training data $T_x^M$ as in Eq.~\eqref{eq:cdata}.
Note that Bob can do so because he has knowledge of the bitstring $x$.
In particular, Bob can generate $f^{(i)}$ from $\{0,1\}^N$, compute $R_{f^{(i)}}(x)$ and pick an element $(y, b)$ from it.
Next, Alice applies the measure protocol $M$ to $\ket{\psi_f}$ for her input $f \in \{0,1\}^N$ and sends $\widehat{\psi}_f = M(\ket{\psi_f})$ to Bob.
Finally, Bob applies $A$ on the data $T_x^M$ he generated and Alice's input $\widehat{\psi}_f$ to obtain a sample $(x, y,b) \sim \widetilde{\pi}_x(\widehat{\psi}_f)$.
Since we assumed that $p_{\mathrm{succ}} >0$, we know that for any $x \in \{0,1\}^n$ there must exist training data $\hat{T}_x^M$ and internal randomization of the learning algorithm $A$ such that the polynomial-time quantum algorithm output by the protocol satisfies Eq.~\eqref{eq:success_criteria_c}.
Throughout the remainder of this proof, we assume Bob fixes this to be the training data and internal randomization he uses for his input $x$ (note that Bob can do so because this does not depend on the input of Alice).
Based on this fixed choice of training data and internal randomization we partition $\{0,1\}^N = \mathrm{F}^x_{\text{good}} \sqcup \mathrm{F}^x_{\text{bad}}$, where $\mathrm{F}^x_{\text{good}}$ denotes the set of functions $f$ for which 
\begin{align}
\label{eq:f_good}
|| \widetilde{\pi}_x(f) - \pi_x(f) ||_{TV} \leq 1 - \epsilon,
\end{align}
where $\widetilde{\pi}_x$ is the random function implemented by the quantum algorithm output by the protocol when using the training data $\hat{T^M_x}$ and internal randomization as above.
Moreover, we note $|\mathrm{F}^x_{\text{good}}| \geq \delta\cdot2^n$ by Eq.~\eqref{eq:success_criteria_c}.
Finally, due to Eq.~\eqref{eq:success_criteria_c} we find that the probability that $(y, b) \in R_f(x)$ is at least
\begin{align}    
\label{eq:prob_support}
\mathsf{Pr}\big((y, b) \in R_f(x) \big) \geq \epsilon
\end{align}
for all $f \in \mathrm{F}^x_{\text{good}}$.
In conclusion, we find that the above described protocol is a \emph{randomized} one-way communication protocol for HM with success probability at least $\epsilon$ for all inputs $(x, f)$ in the subset
\begin{align}   
    \label{eq:x}
    \mathcal{X} := \bigcup_{x \in \{0,1\}^n} \{x\} \times \mathrm{F}_{\text{good}}^x.
\end{align}
In the remainder of our proof, we let $A'(x, \widehat{\psi}_f)$ denote the protocol that Bob runs on his side (i.e., generating the training data $T_x^M$, running the algorithm $A$ on it, and drawing a sample from $\widetilde{\pi}_x(\widehat{\psi}_f)$).
Also, we ensure Bob does so using only classical randomized computation by  classically simulating the quantum algorithms.
Next, we use Yao's principle to show that the above randomized one-way communication protocol implies the existence of a deterministic one-way communication protocol that errs with probability at most $\epsilon\cdot\delta$ with respect to $\mu$ (which would violate Theorem~\ref{thm:hm} since $\epsilon\cdot\delta > 7/8$).
Let $\mathcal{A}$ denote the family of \emph{deterministic} protocols obtained by ``hardwiring'' all possible internal randomizations of the evaluation of $\widetilde{\pi}_x$ by $A'$, i.e.,
\begin{align}
\label{eq:algorithms}
\mathcal{A} = \{A'_r(., .) \mid r \in \{0,1\}^{\mathrm{exp}(n)}\}.
\end{align} 

Also, let $\mathds{X}$ be the random variable with values $(x, f)$ distributed according to the uniform distribution over $\mathcal{X}$, and let $\mathds{A}$ be the random variable over $\mathcal{A}$ where the $r$ is uniformly random.
Finally, we define the function $s: \mathcal{X} \times \mathcal{A} \rightarrow \mathbb{R}$ as
\begin{align}
\label{eq:s}
s( (x, f), A'_r(.,.)) = \mathds{1}\left[A'_r(x, \widehat{\psi}_f) \in R_f(x)\right].
\end{align}

\begin{theorem}[Yao's principle]
Let $\mathds{A}$ be a random variable with values in $\mathcal{A}$ as defined in Eq.~\eqref{eq:algorithms}, and let $\mathds{X}$ be a random variable with values in $\mathcal{X}$ as defined in Eq.~\eqref{eq:x}. 
Then,
\begin{align}
    \label{eq:yao}
    \min_{(x, f) \in \mathcal{X}}\mathbb{E}\left[s((x, f), \mathds{A})\right] \leq \max_{A'_r \in \mathcal{A}}\mathbb{E}\left[s(\mathds{X}, A'_r)\right]
\end{align}
where $s$ is the function defined in Eq.\eqref{eq:s}.  
\end{theorem}

Observe that the quantity $\mathbb{E}\left[s(\mathds{X}, A'_r)\right]$ is precisely the \emph{success probability of the deterministic algorithm}
$A'_r \in \mathcal{A}$ \emph{with respect to the uniform distribution over} $\mathcal{X}$.
Thus, Eq.~\eqref{eq:yao} implies the existence of a deterministic algorithm $A'_r$ such that
\begin{align}
    \label{eq:yao1}
    \mathbb{E}\left[s(\mathds{X}, A'_r)\right] = \mathsf{Pr}_{(x, f) \sim \mathcal{U}(\mathcal{X})}\big(A'_r(x, \widehat{\psi}_f) \in R_f(x)\big) \geq \min_{(x, f) \in \mathcal{X}}\mathbb{E}\left[s((x, f), \mathds{A})\right].
\end{align}

Moreover, observe that the quantity $\min_{(x, f) \in \mathcal{X}}\mathbb{E}\left[s((x, f), \mathds{A})\right]$ is precisely the \emph{success probability of the randomized algorithm $A'$}, which we have previously shown to be at least $\epsilon$.
By combining this with Eq.~\eqref{eq:yao1} we find that
\begin{align}
    \mathsf{Pr}_{(x, f) \sim \mathcal{U}(\mathcal{X})}\big(A'_r(x, \widehat{\psi}_f) \in R_f(x)\big) \geq \epsilon.
\end{align}
Moreover, since $|\mathrm{F}^x_{\text{good}} |\geq \delta\cdot 2^n$ we find that
\begin{align}
    \mathsf{Pr}_{(x, f) \sim \mu}\big(A'_r(x, \widehat{\psi}_f) \in R_f(x)\big) \geq \epsilon \cdot \delta.
\end{align}
Finally, since $\epsilon \cdot \delta > 7/8$, this violates the classical hardness of HM outlined in Theorem~\ref{thm:hm}.

\end{proof}

\subsection{Restricting the input quantum states to pseudorandom phase states}
\label{subsec:unlearnability_pr}

A crucial limitation of the learning problem outlined in Section~\ref{subsec:problem_models} from a pragmatic perspective is that preparing a general phase state is intractable (i.e., not realized by polynomial-time processes). 
In particular, this raises the question of whether separations could persist for states that are prepared by (natural or artificial) polynomial-time processes.
To address this limitation, we show that the concept class in Definition~\ref{eq:concepts} remains not measure-first learnable, even when we constrain the input of the random functions $\pi_x$ to phase states of so-called \textit{pseudorandom functions}. 
Notably, phase states corresponding to appropriately chosen pseudorandom functions can be efficiently prepared. 
Our definition of pseudorandom functions is as follows.

\begin{definition}[Quantum-secure pseudorandom function (QPRF)~\cite{brakerski:pseudo}] \label{def:prf}
    Let $\mathcal{K} = \{\mathcal{K}_n\}_{n \geq 1}$ be an efficiently samplable key distribution, and let $\mathsf{PRF} = \{\mathsf{PRF}_n\}_{n \geq 1}$, $\mathsf{PRF}_n : \mathcal{K}_n \times \{0,1\}^n \rightarrow \{0,1\}$ be an efficiently computable function. 
    We say $\mathsf{PRF}$ is a quantum-secure pseudorandom function if for every efficient non-uniform quantum algorithm $A$ that can make quantum queries there exists a negligible function $\mathrm{negl}(.)$ such that for every $n\geq 1$:
    \begin{align}
        \label{eq:prf}
        \left|\mathsf{Pr}_{k \sim \mathcal{U}(\mathcal{K}_n)}\left[A^{\mathsf{PRF}_n(k, .)}() = 1 \right] - \mathsf{Pr}_{f \sim \mathcal{U}(\{0,1\}^n)}\left[A^{f}() = 1 \right]\right| \leq \mathrm{negl}(n).
    \end{align}
\end{definition}

We remark that if every function $\mathsf{PRF}_n$ admits a classical circuit of size $s(n)$ and depth $d(n)$, then one can prepare the corresponding phase states using a quantum circuit of size $\mathcal{O}(s(n))$ and depth $d(n) + 1$~\cite{brakerski:pseudo}.
Moreover, the existence of such $\mathsf{PRF}_n$ is implied by the existence of quantum secure one-way functions~\cite{zhandry:pr}.

\subsubsection{Fully-quantum learnability with pseudorandom phase states}
Note that when we constrain the inputs of $\pi_x$ to phase states of pseudorandom functions, we essentially modify the distribution over input states in Eq.~\eqref{eq:success_criteria_q} and Eq.~\eqref{eq:success_criteria_c}. 
This new distribution now only has support on phase states that are efficiently preparable.
While Proposition~\ref{prop:fully-quantum_learnable} examines general quantum phase states as input states (which are not typically efficiently preparable), we note that the fully-quantum learnability directly extends the setting where we limit ourselves to efficiently preparable phase states as well.
We summarize this observation in the following proposition (whose proof is the same as that of Proposition~\ref{prop:fully-quantum_learnable}).

\begin{restatable}{proposition}{fullyquantum}
\label{prop:fully-quantum_learnable_ps}
    Let $\mathcal{S}_{\mathrm{pr}} = \{\ket{\psi_{f^{(k)}}}^{\otimes \ell} \mid f^{(k)}(.) = \mathsf{PRF}_n(k,.),\text{ }k \in \mathcal{K}\}$, where $\mathsf{PRF}$ is a quantum-secure pseudorandom function with keys $\mathcal{K}$.
    The concept class in Definition~\ref{def:concepts} is $(0, 0, 1)$-fully-quantum learnable when the distribution over input states is uniform over $\mathcal{S}_{\mathrm{pr}}$.
\end{restatable}

\subsubsection{Limitations of measure-first protocols with pseudorandom phase states}
In the last section, we discussed how the concept class in Definition~\ref{def:concepts} remains fully-quantum learnable when restricted to phase states of pseudorandom functions.
Conversely, in this section we show that this concept class also remains \textit{not} measure-first learnable when restricted to phase states of pseudorandom functions.

\begin{restatable}{theorem}{measurefirstpr}
\label{thm:measure-first_nonlearnable_pr}
    Let $\mathcal{S}_{\mathrm{pr}} = \{\ket{\psi_{f^{(k)}}}^{\otimes \ell} \mid f^{(k)}(.) = \mathsf{PRF}_n(k,.),\text{ }k \in \mathcal{K}\}$, where $\mathsf{PRF}$ is a quantum-secure pseudorandom function with keys $\mathcal{K}$.
    The concept class in Definition~\ref{def:concepts} is not $(\epsilon, \delta, p_{\mathrm{succ}})$-measure-first learnable for $\epsilon \cdot \delta \cdot p_{\mathrm{succ}} > c$ for any constant $c>7/8$ when the distribution over input states is uniform over $\mathcal{S}_{\mathrm{pr}}$.
\end{restatable}

\noindent The main results of this section directly follows from combining Theorem~\ref{thm:measure-first_nonlearnable_pr} and Proposition~\ref{prop:fully-quantum_learnable_ps}.

\begin{corollary}[informal]
  If there exist quantum-secure pseudorandom functions, then there exist a quantum supervised learning problem with efficiently generatable quantum data, which cannot be learned by any measure-first protocol while there exist a fully-quantum protocol which satisfies the learning condition.   
\end{corollary}

The proof of Theorem~\ref{thm:measure-first_nonlearnable_pr} is provided below, and we first present a concise overview of the proof here.
The main idea behind the proof is to illustrate that if the concepts are measure-first learnable when restricted to pseudorandom phase states, then the corresponding measure-first learning protocol can be harnessed to create a non-uniform quantum algorithm that is able to distinguish between truly random functions and pseudorandom functions.
More precisely, this ``distinguisher'' algorithm employs the measure-first learning protocol and evaluates its performance when applied to the phase state corresponding to the function it has been given oracular access to. 
Given that, in the proof of Theorem~\ref{thm:measure-first_nonlearnable}, we have established an upper bound on the generalization performance of any measure-first protocol for truly random phase states, and we assume that the measure-first protocol performs well on pseudorandom phase states, the outcomes of the ``distinguisher'' algorithm should be significantly different depending on whether it is provided oracular access to a truly random or pseudorandom function, contradicting the pseudorandomness assumption.



\begin{proof}[Proof of Theorem~\ref{thm:measure-first_nonlearnable_pr}]

Suppose the concept class in Definition~\ref{def:concepts} is $(\epsilon, \delta, p_{\mathrm{succ}})$-measure-first learnable with $p_{\mathrm{succ}}\cdot \epsilon\cdot\delta > c$ for a constant $c>7/8$ when the distribution over input states is uniform over $\mathcal{S}_{\mathrm{pr}}$ using a measure-first protocol given by $(M, A)$.
That is, for every $\pi_x \in \mathcal{C}$, with probability  at least $p_{\mathrm{succ}}$ we have
\begin{align}
\label{eq:success_criteria_pr}
\mathsf{Pr}_{k \sim \mathcal{U}(\mathcal{K}_n)}\left(|| \widetilde{\pi}_x(f) - \pi_x(f) ||_{TV} \leq 1 - \epsilon \right) \geq \delta,
\end{align}
where $f^{(k)}(.) = \mathsf{PRF}(k, .)$ and $\widetilde{\pi}_x$ is the randomized quantum function obtained from $A$ on input of the form
\begin{align}
    \label{eq:cdata_pr}
        T_x^M = \left\{\left(\widehat{\psi}_{f^{(k)}}, (x, y, b)\right) \; \big| \;(x, y,b)\sim\pi_x(f)\text{ and } k \sim \mathcal{U}\left(\mathcal{K}_n\right)  \right\}_{i=1}^{\text{poly}(n)}.
\end{align}

The main goal of the remainder of the proof is to show that the above assumptions violates the assumption that $\mathsf{PRF}$ is a quantum-secure pseudorandom function.
To achieve this, we devise a quantum algorithm, denoted as $A^f$, which is query access to a function $f$, and which will exhibit a significant difference in the probability of outputting 1 when provided with either a truly random function $f$ or a pseudorandom function $f^{(k)}$. 
In essence, $A^f$ will train a measure-first protocol on phase states of pseudorandom functions and evaluate its performance on the provided function $f$, outputting 1 if it produces a correct sample $(x, y, b)$ with $(y, b) \in R_f(x)$. 
Assuming our measure-first protocol can successfully learn the concepts for phase states of pseudorandom functions, $A^f$ will most likely output 1 when $f$ is pseudorandom. Conversely, if $f$ is truly random, then based on arguments similar to those used in the proof of Theorem~\ref{thm:measure-first_nonlearnable}, the measure-first learning protocol is likely to be incorrect, leading $A^f$ to most of the time output 0.
In particular, we consider the polynomial-time quantum algorithm $A^f$ that does the following:
\begin{enumerate}[label=(\arabic*)]
    \item Sample $x \sim \mathcal{U}(\{0,1\}^n)$.
    \item Generate a set of examples $T_x^M$ as in Eq.~\eqref{eq:cdata_pr}\footnote{Note that we can do so efficiently using a quantum algorithm since we only consider phase states of pseudo-random functions.}.
    \item Use the learning algorithm $A$ with set of examples $T_x^M$to obtain a quantum algorithm $A'$ for $\widetilde{\pi}_x$.
    \item Using quantum query access to $f$ prepare $\ket{\psi_f}^{\otimes \ell}$.\footnote{This step is also efficient both for random and pseudorandom function since we suppose oracle access to $f$.}
    \item Apply $M$ to $\ket{\psi_f}^{\otimes \ell}$ to obtain $\widehat{\psi}_f = M\left(\ket{\psi_f}^{\otimes \ell}\right)$.
    \item Apply $A'$ to $\widehat{\psi}_f$ to obtain a sample $(x, y, b)$ and output 1 if $y \in R_f(x)$, and 0 otherwise.
\end{enumerate}
By the Eq.~\eqref{eq:success_criteria_pr} and the paragraph leading up to it, we know that
\begin{align}
    \label{eq:acc_pr}
    \mathsf{Pr}_{k \sim \mathcal{U}(\mathcal{K}_n)}\left[A^{f^{(k)}} = 1 \right] \geq p_{\mathrm{succ}} \cdot \epsilon\cdot\delta > c.
\end{align}
On the other hand, from the classical lower bound for the HM problem in Theorem~\ref{thm:hm}, we know that 
\begin{align}
    \label{eq:acc_ur}
    \mathsf{Pr}_{f \sim \mathcal{U}(\{0,1\}^N)}\left[A^{f}(.) = 1 \right] \leq 7/8.
\end{align}

\noindent In particular, if Eq.~\eqref{eq:acc_ur} does not hold, then one can construct a one-way communication protocol for HM that succeeds with probability at least $7/8$ with respect to $\mu$ by having Bob perform steps $(2)-(3)$, having Alice perform steps $(4)-(5)$, and sending $\widehat{\psi}_f$ to Bob to perform step (6).
In summary, we conclude that the measure-first protocol, when trained on phase states of pseudorandom functions, cannot generalize well to truly random functions based on the lower-bound established for the HM problem in Theorem~\ref{thm:hm}. 
Moreover, given our assumption that the concept class $\mathcal{C}$ in Definition~\ref{def:concepts} is $(\epsilon, \delta, p_{\text{succ}})$-measure-first learnable on phase states of pseudorandom states, it has to generalize well to other pseudorandom states. 
This implies a distinctive behavior of the ``benchmarking algorithm'' $A^f$ when provided with access to either a pseudorandom function $f^{(k)}$ or a truly random function $f$.
In other words, we thus conclude that Eq.~\eqref{eq:acc_pr} and Eq.~\eqref{eq:acc_ur} are in contradiction with the assumption that $\mathsf{PRF}$ is a quantum-secure pseudorandom function.

\end{proof}

\section{Conclusion}
\label{sec:conclusion}

In our study, we explored the constraints and capabilities of learning from quantum data. 
We established a formal machine learning framework that contrasts two protocols: ``fully quantum'', which adjusts measurements based on data, and ``measure-first'' restricted by fixed initial (though arbitrarily powerful) measurements. 
In particular, we provided an example of a learning problem efficiently solved by a fully-quantum protocol but beyond the capabilities of measure-first protocols.
Moreover, we showed that this persists even when we limit the quantum states from those intractable to prepare to efficiently preparable quantum states. 
These findings underscore the crucial role of processing quantum data in machine learning, revealing scenarios where quantum advantages become evident. 
In particular, they imply that certain learning tasks inherently require the ``exponential capacity'' of quantum states, distinct from classical data. 

\paragraph{Acknowledgements}
VD and CG acknowledge the support of  the Dutch Research Council (NWO/ OCW), as part ofthe Quantum Software Consortium programme (project number 024.003.037). 
This work was supported by the Dutch National Growth Fund (NGF), as part of the Quantum Delta NL programme. 
This publication is also part of the project Divide \& Quantum  (with project number 1389.20.241) of the research programme NWA-ORC which is (partly) financed by the Dutch Research Council (NWO).

\appendix

\section{Proof of Proposition~\ref{prop:fully-quantum_learnable}}
\label{appendix:proof_fully-quantum_learnable}

\fullyquantum*

\begin{proof}

To prove that the concept class $\mathcal{C}$ in Definition~\ref{def:concepts} is fully-quantum learnable we will provide a fully-quantum protocol $A$ that does so successfully.
Suppose we are given training data $T_x$ of the form provided in Eq.~\eqref{eq:qdata}.
Firstly, the fully-quantum protocol $A$ reads out $x$ from one of the examples in $T_x$.
Next, it uses the construction of~\cite{aaronson:qubit_coin} to construct a circuit $U_x$ of size $\mathcal{O}(n)$ in time $\mathcal{O}(n)$ such that when measuring the state $\ket{\phi_{f,x}} = U_x \ket{\psi_f}$ in the computational basis it produces $(y, b) \in \{0,1\}^{n+1}$ such that $b = f(y) \oplus f(x \oplus y)$.
Finally, the learning protocol outputs the description of the POVM measurement
\begin{align}
    \Pi_x = \big\{U_x\ket{z}\bra{z}U_x^\dagger \text{  }\large|\text{  } z\in \{0,1\}^n\big\}
\end{align}
as by the above measuring $\Pi_x$ on an arbitrary phase state $\ket{\psi_f}$ implements $\pi_x$ with zero error.

While it may appear that little learning is occurring when we can readily extract $x$ from a single example in $T_x$, we can introduce varying degrees of learning by not appending the complete description of $x$ to the examples. 
Instead, we include only partial information about $x$ that still allow us to recover a full description of $x$ using a polynomial number of examples.
For instance, instead of appending $x$ to the examples we can append certain functions $g_i(x)$, where $g_i$ is drawn uniformly random from some set of  $\mathcal{G} = \{g_i\}_{i \in I}$.
For instance, for $i \in \{0,1\}^n$ we can consider functions like 
\begin{align}
\label{eq:cdot}
g_i(x) = \left(i\cdot x, i\right) \in \{0,1\}^{n+1},
\end{align}
where $x\cdot i = \sum_{j = 1}^n x_j \cdot i_j \mod 2$.
Another example of such a family of functions would be
\begin{align}
\label{eq:dlp}
g_i(x) = \left(\mathrm{DLP}(x)_i, i\right) \in \{0,1\}^{n+1},
\end{align}
where $\mathrm{DLP}_i(x)$ denotes the $i$th bit of the discrete logarithm of $x$ in a suitably chosen group.
For these functions, one can show that $x$ can be recovered with high probability from a polynomial number of evaluations of $g_i(x)$ for randomly chosen $g_i$ from $\mathcal{G}$. 
Moreover, functions similar to the $g_i$ in Eq.~\eqref{eq:dlp} require a quantum computer to be able to efficiently recover $x$~\cite{liu:dlp}.

\end{proof}

\bibliographystyle{unsrt}
\bibliography{main}

\begin{thebibliography}{10}

\bibitem{chen:memory}
Sitan Chen, Jordan Cotler, Hsin-Yuan Huang, and Jerry Li.
\newblock Exponential separations between learning with and without quantum
  memory.
\newblock In {\em 2021 IEEE 62nd Annual Symposium on Foundations of Computer
  Science (FOCS)}, 2022.

\bibitem{huang:science}
Hsin-Yuan Huang, Richard Kueng, Giacomo Torlai, Victor~V Albert, and John
  Preskill.
\newblock Provably efficient machine learning for quantum many-body problems.
\newblock {\em Science}, 377, 2022.

\bibitem{huang:nature}
Hsin-Yuan Huang, Richard Kueng, and John Preskill.
\newblock Predicting many properties of a quantum system from very few
  measurements.
\newblock {\em Nature Physics}, 16, 2020.

\bibitem{elben:review}
Andreas Elben, Steven~T Flammia, Hsin-Yuan Huang, Richard Kueng, John Preskill,
  Beno{\^\i}t Vermersch, and Peter Zoller.
\newblock The randomized measurement toolbox.
\newblock {\em Nature Reviews Physics}, 5, 2023.

\bibitem{huang:classicalshadows}
Hsin-Yuan Huang, Richard Kueng, and John Preskill.
\newblock Predicting many properties of a quantum system from very few
  measurements.
\newblock {\em Nature Physics}, 16(10):1050--1057, 2020.

\bibitem{montanaro:fourier}
Ashley Montanaro.
\newblock Quantum states cannot be transmitted efficiently classically.
\newblock {\em Quantum}, 3, 2019.

\bibitem{aaronson:qubit_coin}
Scott Aaronson, Harry Buhrman, and William Kretschmer.
\newblock A qubit, a coin, and an advice string walk into a relational problem.
\newblock {\em {\normalfont \texttt{arXiv:2302.10332}}}, 2023.

\bibitem{aimeur:quantum_data}
Esma A{\"\i}meur, Gilles Brassard, and S{\'e}bastien Gambs.
\newblock Machine learning in a quantum world.
\newblock In {\em Proceedings of the Advances in Artificial Intelligence: 19th
  Conference of the Canadian Society for Computational Studies of
  Intelligence}, pages 431--442. Springer, 2006.

\bibitem{aaronson:shadowtomography}
Scott Aaronson.
\newblock Shadow tomography of quantum states.
\newblock In {\em Proceedings of the 50th annual ACM SIGACT symposium on theory
  of computing}, pages 325--338, 2018.

\bibitem{chen:exponential}
Sitan Chen, Jordan Cotler, Hsin-Yuan Huang, and Jerry Li.
\newblock Exponential separations between learning with and without quantum
  memory.
\newblock In {\em 2021 IEEE 62nd Annual Symposium on Foundations of Computer
  Science (FOCS)}, pages 574--585. IEEE, 2022.

\bibitem{aaronson:distributions}
Weiyuan Gong and Scott Aaronson.
\newblock Learning distributions over quantum measurement outcomes.
\newblock In {\em International Conference on Machine Learning}, pages
  11598--11613. PMLR, 2023.

\bibitem{sweke:learnability}
Ryan Sweke, Jean-Pierre Seifert, Dominik Hangleiter, and Jens Eisert.
\newblock On the quantum versus classical learnability of discrete
  distributions.
\newblock {\em Quantum}, 5:417, 2021.

\bibitem{cheng:learnability}
Hao-Chung Cheng, Min-Hsiu Hsieh, and Ping-Cheng Yeh.
\newblock The learnability of unknown quantum measurements.
\newblock {\em {\normalfont \texttt{arXiv:1501.00559}}}, 2015.

\bibitem{jerbi:power}
Sofiene Jerbi, Joe Gibbs, Manuel~S Rudolph, Matthias~C Caro, Patrick~J Coles,
  Hsin-Yuan Huang, and Zo{\"e} Holmes.
\newblock The power and limitations of learning quantum dynamics incoherently.
\newblock {\em {\normalfont \texttt{arXiv:2303.12834}}}, 2023.

\bibitem{valiant:pac}
Leslie~G Valiant.
\newblock A theory of the learnable.
\newblock {\em Communications of the ACM}, 27(11):1134--1142, 1984.

\bibitem{bar:hm}
Ziv Bar-Yossef, Thathachar~S Jayram, and Iordanis Kerenidis.
\newblock Exponential separation of quantum and classical one-way communication
  complexity.
\newblock In {\em Proceedings of the thirty-sixth annual ACM symposium on
  Theory of computing}, 2004.

\bibitem{brakerski:pseudo}
Zvika Brakerski and Omri Shmueli.
\newblock (pseudo) random quantum states with binary phase.
\newblock In {\em Theory of Cryptography Conference}. Springer, 2019.

\bibitem{zhandry:pr}
Mark Zhandry.
\newblock How to construct quantum random functions.
\newblock In {\em 2012 IEEE 53rd Annual Symposium on Foundations of Computer
  Science}. IEEE, 2012.

\bibitem{liu:dlp}
Yunchao Liu, Srinivasan Arunachalam, and Kristan Temme.
\newblock A rigorous and robust quantum speed-up in supervised machine
  learning.
\newblock {\em Nature Physics}, 2021.

\end{thebibliography}

\end{document}